\documentclass[11pt]{article}

\usepackage{amsthm}
\newtheorem{theorem}{Theorem}[section]
\newtheorem{lemma}[theorem]{Lemma}
\newtheorem{corollary}[theorem]{Corollary}

\theoremstyle{definition}
\newtheorem{definition}[theorem]{Definition}

\usepackage{microtype}
\usepackage{amsmath}
\usepackage{amsfonts}
\usepackage{amssymb}
\usepackage[round,authoryear]{natbib}
\allowdisplaybreaks

\usepackage[dvipsnames]{xcolor}
\usepackage{hyperref}
\usepackage{euscript}

\hypersetup{
  colorlinks=true,
  linkcolor=blue,
  citecolor=blue,
  urlcolor=blue,
}
\usepackage{algorithm}
\usepackage{algpseudocode}
\algtext*{EndFor}
\newcommand{\R}{\mathbf{R}}
\newcommand{\opt}{\mathrm{OPT}}
\renewcommand{\phi}{\varphi}
\newcommand{\eps}{\varepsilon}
\newcommand{\E}{\mathbf{E}}
\newcommand{\alg}{Algorithm~\ref{alg:noisy} }
\renewcommand{\H}{\EuScript{H}}
\newcommand{\U}{\EuScript{U}}
\newcommand{\F}{\EuScript{F}}
\newcommand{\A}{\EuScript{A}}

\newcommand{\one}{\mathbf{1}}
\newcommand{\high}{\mathrm{high}}
\newcommand{\midd}{\mathrm{mid}}
\newcommand{\low}{\mathrm{low}}
\newcommand{\Ev}{\EuScript{E}}
\renewcommand{\S}{\EuScript{S}}

\title{Noisy \(k\)-means++ is Not \emph{too} Noisy}

\author{
  Poojan Shah\\
 \small{ Department of Computer Science and Engineering}\\
  \small{Indian Institute of Technology Delhi}\\
  \small{\texttt{shahpoojan2004@gmail.com}}
}

\date{}

\begin{document}

\maketitle

\begin{abstract} \noindent \setlength{\parindent}{0pt}
The celebrated \(k\)-means++ algorithm of \citeauthor{av2007}~(SODA~\citeyear{av2007}) achieves an \(O(\log k)\) expected approximation guarantee for the classical \(k\)-means problem using a technique called \(D^2\)-sampling, which has become ubiquitous in the design of clustering algorithms. \citeauthor{behs2020}~(ESA~\citeyear{behs2020}) introduced \(\eps\)-noisy \(k\)-means++, in which the sampling probabilities may incur an adversarial multiplicative error of $(1\pm\eps)$, but were only able to recover an \(O(\log^2 k)\) approximation guarantee. Recently, \citeauthor{gor2023}~(ESA~\citeyear{gor2023}) recovered the asymptotic \(O(\log k)\) approximation, but their analysis loses a constant factor of approximately \(147{,}638\) even as \(\varepsilon\to0\), leaving open the possibility that \(k\)-means++ may be highly sensitive to even small amounts of adversarial noise. They asked whether one can instead obtain a bound within a \(1+O(\varepsilon)\) factor of the classical guarantee. We resolve this question affirmatively, proving an expected approximation guarantee of
\[
8(\ln k+2)\left(\frac{1+\varepsilon}{1-\varepsilon}\right)^4
=(1+O(\varepsilon))\,8(\ln k+2).
\]
The main technical obstacle is controlling the average cost of the \emph{uncovered} clusters in the potential-based framework of \cite{dasgupta2013}. The proof goes through the analysis of a simple to state \emph{adversarial sampling game}; where the goal is to obtain a tight bound on the average weights of elements as the game progresses. We overcome this obstacle with a short and elementary argument by considering the complete cumulative distribution function of the weights and showing that a certain high-low statistic related to the average weight is a super-martingale.  \\ 

We complement the upper bound with two separations. First, a noisy version of the \cite{av2007} lower-bound instance incurs a \(1+\Omega(\varepsilon)\) loss over exact \(k\)-means++, showing that linear dependence on the noise is necessary. Second, pointwise multiplicative control is qualitatively essential: if it is replaced
by per-round total variation closeness, then no finite approximation
guarantee is possible, even for \(k=2\).
\end{abstract}

\section{Introduction}

The \(k\)-means problem is a classical problem in computer science. Given a set of \(n\) points \(X \subset \R^d\) and an integer \(k \geq 1\), the goal is to output \(k\) centers \(C \subset \R^d\) minimizing the cost function  
\[
\phi(X,C)=\sum_{x\in X}\min_{c\in C}\|x-c\|^2.
\]
We write \(\opt_k(X)=\min_{C\subset\R^d:\,|C|=k}\phi(X,C)\) for the optimal $k$-means cost. The celebrated \(k\)-means++ algorithm of \cite{av2007} uses \(D^2\)-sampling to obtain an \(O(\log k)\) approximation in expectation. It samples the first center uniformly from \(X\). Thereafter, given the centers \(C_j\) already chosen, it samples \(x\in X\) with probability
\[
p_{C_j}(x)=\frac{\phi(x,C_j)}{\phi(X,C_j)},
\]
where \(\phi(x,C_j):=\phi(\{x\},C_j)\), and adds \(x\) to \(C_j\): $C_{j+1} = C_j \cup \{x\}$.\\

The practical influence of this idea extends beyond its theoretical applications: for example, the \texttt{KMeans} implementation in scikit-learn~\citep{pedregosa2011scikit,sklearnKMeans} uses a greedy variant of \(k\)-means++ as its default initialization. The classical analysis, however, assumes that every \(D^2\)-sampling probability is computed and sampled exactly. This assumption is stronger than what an implementation can generally guarantee. Squared distances, their sum, and the resulting normalized probabilities are evaluated using finite-precision arithmetic and are therefore subject to rounding error~\citep{goldberg1991}. Approximate distance computations, reduced-precision arithmetic, and other performance-oriented implementations may introduce further perturbations. Thus, even when the intended distribution is the \(D^2\)-distribution, the distribution used by an implementation may only approximate it.\\

Motivated by this discrepancy, \cite{behs2020} initiated the study of \(\eps\)-\emph{noisy} \(k\)-means++; see Algorithm~\ref{alg:noisy}. At each step \(j\), an adaptive adversary chooses a normalized distribution \(q_j\) satisfying
\[
(1-\eps)p_{C_j}(x)\leq q_j(x)\leq(1+\eps)p_{C_j}(x)
\qquad\text{for every }x\in X.
\]
from which the next center is then sampled from. Intriguingly, the classical analysis of \cite{av2007} does not extend directly to this scenario even for a small constant amount of noise such as \(\eps=0.01\). Using a much more involved analysis, \cite{behs2020} obtained an \(O(\log^2 k)\) guarantee. Subsequently, \cite{gor2023} recovered an asymptotic \(O(\log k)\) guarantee. The latter analysis, however, loses a constant factor of approximately
\(\frac{90}{(1-e^{-1/40})^2}\approx147{,}638\) even as \(\eps\to0\). Thus, \cite{gor2023} ask whether the approximation ratio can be brought within a \(1+O(\varepsilon)\) factor of the classical \(k\)-means++ analysis, or whether a counterexample rules this out. In this work, we answer this question in the affirmative with our main result: 

\begin{theorem}[Main Theorem] \label{thm:main}
For every nonempty finite set \(X\subset\R^d\), integer \(1\leq k\leq |X|\), and noise parameter \(0\leq\eps<1\), let \(C\) be the centers returned by Algorithm~\ref{alg:noisy}. Against every adaptive adversary satisfying the noisy sampling rule,
\[
\E[\phi(X,C)]
\leq
8(\ln k+2)
\left(\frac{1+\eps}{1-\eps}\right)^4
\opt_k(X).
\]
For \(0\leq\eps<1/2\), the additional multiplicative factor is at most \(1+160\eps\).
\end{theorem}

\paragraph{Remark on the leading constant.} We state and prove Theorem 1.1 using the classical covered-cluster bound of \cite{av2007}, as presented in the potential analysis of \cite{dasgupta2013}. Later, the leading constant in the approximation guarantee was sharpened from $8$ to $5$ by \cite{makarychev2020}. Their refinement is orthogonal to our new argument: conditioned on sampling from a fixed uncovered optimal cluster, the noisy conditional distribution is pointwise within a factor \((1+\varepsilon)/(1-\varepsilon)\) of the corresponding exact \(D^2\)-distribution. Consequently, substituting their covered-cluster lemma into our analysis gives

\[
\E[\phi(X,C)]
\le
5(\ln k+2)
\left(\frac{1+\varepsilon}{1-\varepsilon}\right)^4
\operatorname{OPT}_k(X).
\]

We retain the constant \(8\) in the remainder of the paper to keep the exposition aligned with the preceding analyses and to isolate the new ingredient controlling the uncovered clusters.

\subsection{Discussion on our techniques and related work }

The analysis of noisy $k$-means++ in prior work usually goes through a related \emph{adversarial sampling game} the details of which we discuss in Section~\ref{sec:sampling-game}. Indeed, our improved bounds for noisy $k$-means++ are also a consequence of a more refined analysis of this sampling game in Section~\ref{sec:improved-sampling-game}. A high level description of the game is as follows: initially, there is a set of elements, with each being assigned a non-negative weight value. The game then proceeds in successive rounds. In each round, an adversary can either sample an element according to a distribution which is \emph{close} to being proportional to its weight and delete it from the set of elements, followed by decreasing the weights of all surviving elements. Our goal is to study how the expected weight of the surviving elements evolves over the rounds. Ideally, we want the average to grow in a bounded manner as compared to the average at the start of the game. A tighter upper bound on the expected average through the rounds with respect to the initial average corresponds to a better approximation guarantee for Algorithm~\ref{alg:noisy}. \\ 

\cite{behs2020}  bound the average via a two-stage probabilistic argument - a Chernoff bound to show that most deletions behave like proportional sampling, followed by a union bound over all surviving elements to control the maximum surviving value, which yields a bound with a residual $O(\log k )$ factor. \cite{gor2023} instead partition elements into fixed big/medium/small weight classes and analyze the transitions between these classes directly, which removes the 
$\log k$ dependence but at the cost of an enormous, non-vanishing constant. Our key insight is to study the \emph{cumulative distribution function} of the weights as a whole instead of directly attacking the average itself. Instead of a coarse and fixed  partitioning of the weights, we consider a high-mid-low partition at every threshold $t \geq 0$ simultaneously (Definition~\ref{def:decomposition}).  We then  show that a high–low ratio statistic related to the CDF is a bounded supermartingale (Lemma~\ref{lem:high-low-super-martingale}), a fact that follows from one direct pairwise probability comparison (Lemma~\ref{lem:high-low-prob}) rather than concentration inequalities or case analysis. Integrating this one inequality over all thresholds via the layer-cake identity (Lemma~\ref{lem:tail}) then recovers the average cost directly. Because the argument is a martingale inequality rather than a probabilistic tail bound, it holds exactly for every $k$ with no need for $k$ to be large, and because it uses a continuum of thresholds rather than a fixed partition, it avoids both the log-factor loss of \cite{behs2020} and the non-vanishing constant of \cite{gor2023}, yielding the clean 
dependence that converges to the noiseless bound as $\eps \to 0$ in Theorem~\ref{thm:improved-sampling}. \\

Theorems~\ref{thm:linear-noise-lower-bound} and~\ref{thm:tv-no-guarantee} clarify the role of the noise model. The first adapts the lower-bound instance of \cite{av2007} and proves that a \(1+\Omega(\eps)\) loss relative to exact \(k\)-means++ is unavoidable. The second shows that pointwise multiplicative control is qualitatively essential: under per-round total variation closeness, no finite approximation guarantee is possible even for \(k=2\). It is worth noting that there are several algorithms which speed up $k$-means++ in practice using approximate sampling techniques such as MCMC sampling \citep{bachem2016fast,bachem2016approximate} or rejection sampling \citep{cohenaddad2020fast,shah2026fast} which usually have an additive guarantee related to the variance of the dataset stemming from using weaker notion of closeness of distributions (such as TV distance) instead of the stronger pointwise error. This is usually not a problem in practice when $\frac{\opt_1(X)}{\opt_k(X)}$ is bounded, and hence Theorem~\ref{thm:tv-no-guarantee} may be seen as exploiting this weakness. 


\begin{algorithm}[t]
\caption{$\eps$-\emph{noisy} $k$-means++} \label{alg:noisy}
\begin{algorithmic}[1]
\Require \(X \subset \R^d\), \(k \geq 2\), and \(0\leq\varepsilon<1\)
\State Choose a normalized distribution \(q_1\) on \(X\) satisfying \(q_1(x)\in\left[(1-\varepsilon)/n,(1+\varepsilon)/n\right]\) for every \(x\in X\); sample \(x\sim q_1\) and set \(C_1=\{x\}\).
\For{$j \leftarrow 1, \ldots, k-1$}
    \State Choose a normalized distribution \(q_{j+1}\) on \(X\) satisfying
    \[
    q_{j+1}(x)\in
    \left[
    (1-\varepsilon)\frac{\varphi(x,C_j)}{\varphi(X,C_j)},
    (1+\varepsilon)\frac{\varphi(x,C_j)}{\varphi(X,C_j)}
    \right] \quad \forall x \in X
    \]
    \State Sample \(x\sim q_{j+1}\) and set \(C_{j+1}=C_j\cup\{x\}\).
\EndFor
\State \Return $C := C_k$
\end{algorithmic}
\end{algorithm}

\section{The Sampling Game} \label{sec:sampling-game}

The analysis of \cite{behs2020} passes through an abstract adversarial deletion experiment. Motivated by this abstraction, we isolate a closely related \(\beta\)-sampling game that captures precisely the property of noisy \(D^2\)-sampling needed here. We first analyse this game without using any \(k\)-means terminology and then apply the resulting theorem in Section~\ref{sec:potential-analysis}. 

\begin{definition}[$\beta$-sampling game] \label{def:sampling-game}
Let \(\beta\geq1\). Initially, there is a set \(\S_0\) of \(m\) elements. Each \(A\in\S_0\) has a nonnegative weight \(w_0(A)\). Write
\[
w_0(\S_0):=\sum_{A\in\S_0}w_0(A)
\qquad\text{and}\qquad
\rho_0:=\frac{w_0(\S_0)}{|\S_0|}.
\]
The game proceeds in rounds during which elements may be deleted and their weights may decrease.

In round \(j\), let \(\S_j\) be the set of surviving elements and let \(w_j:\S_j\to\R_{\geq0}\) be their weight function. Their average weight is
\[
\rho_j:=\frac{w_j(\S_j)}{|\S_j|}.
\]

After observing the state through round $j$, the adversary chooses an $\Ev_j$-measurable deletion probability $d_{j+1}\in[0,1]$. A fresh Bernoulli random variable $D_{j+1}$ is then drawn with $\Pr(D_{j+1}=1\mid  \Ev_j)=d_{j+1}$.

\begin{enumerate}
    \item If \(D_{j+1}=0\), no element is deleted: \(\S_{j+1}=\S_j\). The adversary chooses a new weight function satisfying
    \[
    0\leq w_{j+1}(A)\leq w_j(A)
    \qquad\text{for every }A\in\S_j.
    \]
    \item If \(D_{j+1}=1\), the adversary chooses a normalized probability distribution \(\pi_j\) on \(\S_j\) satisfying
    \[
    {\beta^{-1}} \cdot \frac{w_j(A)}{w_j(\S_j)}
    \leq \pi_j(A)\leq
    \beta \cdot \frac{w_j(A)}{w_j(\S_j)}
    \qquad\text{for every }A\in\S_j.
    \]
    It samples \(A^*\sim\pi_j\), deletes \(A^*\), and sets \(\S_{j+1}=\S_j\setminus\{A^*\}\). Finally, it chooses a new weight function satisfying
    \[
    0\leq w_{j+1}(A)\leq w_j(A)
    \qquad\text{for every }A\in\S_{j+1}.
    \]
\end{enumerate}

Let \(\Ev_j\) denote the \(\sigma\)-algebra generated by the complete history through round \(j\), including the adversary's past decisions and randomness. The choice of \(d_{j+1}\), and of \(\pi_j\) when \(D_{j+1}=1\), is made from this revealed history and is therefore \(\Ev_j\)-measurable; any internal randomness used to make these choices is included before \(A^*\) is sampled. In particular, \(\S_j\), \(w_j\), and \(\rho_j\) are \(\Ev_j\)-measurable. If \(\S_j=\varnothing\) or \(w_j(\S_j)=0\), the game stops and we set \(\rho_\ell=0\) for every \(\ell\geq j\).
\end{definition}

We first analyse the noiseless case \(\beta=1\) as done by \cite{dasgupta2013}. This calculation identifies the precise obstruction that appears when \(\beta>1\).
\subsection{Analyzing \(\beta=1\)}
We recall the definition properties of a supermartingale; see, for example, \cite{williams1991martingales}.

\begin{definition}[Supermartingale]
Let \((\Omega,\EuScript{F},\mathbf{P})\) be a probability space with a filtration \((\EuScript{E}_j)_{j\geq0}\), that is, an increasing sequence of \(\sigma\)-algebras \(\EuScript{E}_0\subseteq\EuScript{E}_1\subseteq\cdots\subseteq\EuScript{F}\). A sequence of random variables \((\rho_j)_{j\geq0}\) is a \emph{supermartingale} with respect to this filtration if, for every \(j\geq0\),
\begin{enumerate}
\item[(i)] \(\rho_j\) is \(\EuScript{E}_j\)-measurable;
\item[(ii)] \(\mathbf{E}[|\rho_j|]<\infty\);
\item[(iii)] \(\mathbf{E}[\rho_{j+1}\mid\EuScript{E}_j]\leq\rho_j\) almost surely.
\end{enumerate}
The sequence is \emph{bounded} if there exists \(M<\infty\) such that \(|\rho_j|\leq M\) almost surely for all \(j\).
\end{definition}

\begin{lemma}[Tower property of bounded supermartingales]
\label{lem:tower}
Let \((\rho_j)_{j\geq0}\) be a bounded supermartingale with respect to \((\EuScript{E}_j)_{j\geq0}\). Then, for every \(0\leq s<k\),
\[
\mathbf{E}[\rho_k \mid \EuScript{E}_s] \leq \rho_s.
\]
\end{lemma}

For the case of $\beta = 1$, it is easy to see that the expected average weight cannot increase from one round to the next. 

\begin{lemma} \label{lem:beta-one}
In the \(\beta\)-sampling game with \(\beta=1\), \((\rho_j)_{j\geq0}\) is a bounded supermartingale with respect to \((\Ev_j)_{j\geq0}\). Consequently, for every \(0\leq s<k\),
\[
\E[\rho_k\mid\Ev_s]\leq\rho_s.
\]
\end{lemma}
\begin{proof}
    Condition on \(D_{j+1}\). If \(D_{j+1}=0\), then \(\S_{j+1}=\S_j\) and the weights can only decrease, so
    \[
    \E[\rho_{j+1}\mid\Ev_j,D_{j+1}=0]\leq\rho_j.
    \]
    Now suppose \(D_{j+1}=1\). Exactly one element \(A^*\in\S_j\) is deleted. If \(|\S_j|=0\), both \(\rho_j\) and \(\rho_{j+1}\) are \(0\) by definition. If \(|\S_j|=1\), then \(\rho_{j+1}=0\leq\rho_j\). We may therefore assume that \(|\S_j|>1\).

    Let \(\widetilde{\rho}_{j+1}\) be the average immediately after the deletion and before the remaining weights decrease. Then \(\rho_{j+1}\leq\widetilde{\rho}_{j+1}\), and
    \begin{align*}
        &\E[\rho_{j+1}| \Ev_j , D_{j+1} = 1] \leq \E[\widetilde{\rho}_{j+1} | \Ev_j,D_{j+1} = 1] \\ =& \E \left[ \frac{\sum_{A \in \S_{j+1}}w_j(A)}{|\S_{j+1}|} \Big| \Ev_j,D_{j+1} = 1 \right] \\ 
        =& \E \left[  \frac{\sum_{A \in \S_j}w_j(A) - w_j(A^*)}{|\S_j|-1}\Big|\Ev_j, D_{j+1}=1\right] \\ =& \rho_j \left(1 + \frac{1}{|\S_j|-1} \right) - \frac{1}{|\S_j|-1} \E[w_j(A^*) | \Ev_j,D_{j+1}=1].
    \end{align*}

    Let us compute the remaining expectation:
    \[
    \E[w_j(A^*)\mid\Ev_j,D_{j+1}=1]
    =\sum_{A\in\S_j}w_j(A)\pi_j(A)
    =\frac{\sum_{A\in\S_j}w_j(A)^2}{w_j(\S_j)}.
    \]
    By Cauchy--Schwarz,\footnote{Apply
    \(
    (\sum_{A\in\S_j}w_j(A))^2
    \leq
    |\S_j|\sum_{A\in\S_j}w_j(A)^2.
    \)}
    \[
    \sum_{A\in\S_j}w_j(A)^2
    \geq\frac{w_j(\S_j)^2}{|\S_j|}.
    \]
    Therefore,
    \[
    \E[w_j(A^*)\mid\Ev_j,D_{j+1}=1]
    \geq\frac{w_j(\S_j)}{|\S_j|}
    =\rho_j,
    \]
    and hence
    \[
    \E[\rho_{j+1}\mid\Ev_j,D_{j+1}=1]
    \leq
    \rho_j\left(1+\frac{1}{|\S_j|-1}\right)
    -\frac{\rho_j}{|\S_j|-1}
    =\rho_j.
    \]

    \begin{align*}
        \E[\rho_{j+1}|\Ev_j] &= \E[\rho_{j+1}|\Ev_j , D_{j+1} = 0] \cdot \Pr[D_{j+1}=0|\Ev_j] + \E[\rho_{j+1}|\Ev_j , D_{j+1} = 1] \cdot \Pr[D_{j+1}=1|\Ev_j]\\ 
        &\leq \rho_j \left( \Pr[D_{j+1}=0|\Ev_j]+\Pr[D_{j+1}=1|\Ev_j]\right) = \rho_j.
    \end{align*}
\end{proof}

\subsection{What goes wrong when $\beta>1$?}

Let us try to repeat the calculation in Lemma~\ref{lem:beta-one} for \(\beta>1\). Fix a round \(j\), condition on \(\Ev_j\), and write \(n_j:=|\S_j|\). Let \(A^\star\) denote the element removed when \(D_{j+1}=1\). Whenever \(\Pr(D_{j+1}=1\mid\Ev_j)>0\), we have
\begin{align*}
\E\!\left[w_j(A^\star)\mid\Ev_j,D_{j+1}=1\right]
=
\sum_{A\in\S_j}\pi_j(A)w_j(A)
\geq
\frac{1}{\beta\,w_j(\S_j)}
\sum_{A\in\S_j}w_j(A)^2
\geq
\frac{1}{\beta\,w_j(\S_j)}
\cdot
\frac{w_j(\S_j)^2}{n_j}
=
\frac{\rho_j}{\beta}.
\end{align*}
The second inequality follows from Cauchy--Schwarz. Suppose first that \(n_j\geq2\). Since the weights of the surviving elements can only decrease, on the event \(D_{j+1}=1\) we have
\[
\rho_{j+1}
\leq
\frac{w_j(\S_j)-w_j(A^\star)}{n_j-1}.
\]
Taking conditional expectations and using \(w_j(\S_j)=n_j\rho_j\), we obtain
\begin{align*}
\E[\rho_{j+1}\mid\Ev_j,D_{j+1}=1]
\leq
\frac{w_j(\S_j)-
\E[w_j(A^\star)\mid\Ev_j,D_{j+1}=1]}
{n_j-1}
\leq
\frac{n_j\rho_j-\rho_j/\beta}{n_j-1}
=
\left(1+\frac{1-\beta^{-1}}{n_j-1}\right)\rho_j.
\end{align*}
If \(D_{j+1}=0\), then \(\S_j\) remains unchanged and its weights can only decrease. Hence, \(\rho_{j+1}\leq\rho_j\). The case \(n_j=1\) is immediate: if the unique element is removed, the new average is \(0\), and otherwise the average cannot increase. This calculation shows exactly where the proof for \(\beta=1\) breaks. When \(\beta=1\), the multiplier above is \(1\), and \((\rho_j)_{j\geq0}\) is a supermartingale. When \(\beta>1\), however,
\[
1+\frac{1-\beta^{-1}}{n_j-1}>1,
\]
so the same argument no longer shows that the average weight is non-increasing in expectation.
Let \(\alpha:=1-\beta^{-1}\)
and define
\[
G(0)=G(1):=1,
\qquad
G(n):=\prod_{m=1}^{n-1}\left(1+\frac{\alpha}{m}\right)
\quad\text{for }n\geq2.
\]

We claim that \((G(n_j)\rho_j)_{j\geq0}\) is a supermartingale. To see this, fix \(j\), condition on \(\Ev_j\), and let
\[
p_j:=\Pr(D_{j+1}=1\mid\Ev_j).
\]
When \(n_j\geq2\), the preceding bounds give
\begin{align*}
&\E[G(n_{j+1})\rho_{j+1}\mid\Ev_j]\\
&=
p_jG(n_j-1)
\E[\rho_{j+1}\mid\Ev_j,D_{j+1}=1]
+
(1-p_j)G(n_j)
\E[\rho_{j+1}\mid\Ev_j,D_{j+1}=0]\\
&\leq
p_jG(n_j-1)
\left(1+\frac{\alpha}{n_j-1}\right)\rho_j
+
(1-p_j)G(n_j)\rho_j\\
&=
p_jG(n_j)\rho_j
+
(1-p_j)G(n_j)\rho_j =
G(n_j)\rho_j.
\end{align*}
The same conclusion is immediate when \(n_j\leq1\). Thus \(G(n_j)\rho_j\) is indeed a supermartingale.

Applying the tower property from round \(s\) to round \(k\), we obtain
\[
\E[G(n_k)\rho_k\mid\Ev_s]
\leq
G(n_s)\rho_s.
\]
Since \(G(n_k)\geq1\), it follows that
\[
\E[\rho_k\mid\Ev_s]
\leq
G(n_s)\rho_s.
\]
It remains to bound \(G(n_s)\). Using \(1+x\leq e^x\), we have
\begin{align*}
G(n_s)
=
\prod_{m=1}^{n_s-1}\left(1+\frac{\alpha}{m}\right) \leq
\exp\left(\alpha\sum_{m=1}^{n_s-1}\frac1m\right)\leq
\exp\left(\alpha(1+\ln n_s)\right)=
e^{\alpha} n_s^\alpha.
\end{align*}
Since \(n_s\leq k\) and \(\alpha=1-\beta^{-1}\), this proves
\[
\E[\rho_k\mid\Ev_s]
\leq
e^{1-\beta^{-1}}
k^{\,1-\beta^{-1}}\rho_s.
\]

Thus, a more careful version of the direct argument gives a polynomial dependence on \(k\). Nevertheless, this bound still grows with \(k\) for every fixed \(\beta>1\). \cite{behs2020} obtained the stronger bound
\[
\E[\rho_k\mid\Ev_s]\leq O(\log k)\rho_s
\]
for sufficiently large \(k\), while \cite{gor2023} obtained a constant bound of approximately
\[
\E[\rho_k\mid\Ev_s]\lesssim147{,}638\,\rho_s.
\]
Neither bound recovers the noiseless guarantee as \(\beta\to1\). Our main technical contribution, stated in the next section, removes the dependence on \(k\) while retaining the correct limiting behavior:
\[
\E[\rho_k\mid\Ev_s]\leq\beta^2\rho_s.
\] 

\section{Improved Analysis of the Sampling Game} \label{sec:improved-sampling-game}

We now present our improved analysis of the \(\beta\)-sampling game.

\begin{theorem}\label{thm:improved-sampling}

Consider the \(\beta\)-sampling game of Section~\ref{sec:sampling-game}, and let \(\rho_j\) be the average weight in round \(j\). For every adversary and every \(0\leq s<k\),
\[
\E[\rho_k\mid\Ev_s]\leq\beta^2\rho_s.
\]

\end{theorem}

\paragraph{Proof idea.}
For \(\beta>1\), the sequence \((\rho_j)\) need not be a supermartingale. Instead, for every continuous threshold \(t\geq0\), we partition the elements into \emph{low}, \emph{middle}, and \emph{high} classes. We show that the fraction \(Z_j(t)\) of high elements among the low and high elements is a bounded supermartingale. We then compare this statistic with the tails of the empirical weight distribution and integrate the comparison to recover the average weight.

\subsection{Dynamic threshold classes} 

Our first step is to show that the fraction of high-weight elements among the high-weight and low-weight elements cannot increase in expectation.

\begin{definition} \label{def:decomposition}

Consider the elements \(\S_j\) in any round \(j\geq0\) of the \(\beta\)-sampling game, with weights \(w_j:\S_j\to\R_{\geq0}\). For every threshold \(t\geq0\), define
\begin{align*}
    \A^\high_j(t) &:= \{A\in\S_j:w_j(A)\geq\beta^2t\},\\
    \A^\midd_j(t) &:= \{A\in\S_j:t\leq w_j(A)<\beta^2t\},\\
    \A^\low_j(t) &:= \{A\in\S_j:0\leq w_j(A)<t\}.
\end{align*}
Let
\[
h_j(t):=|\A^\high_j(t)|,\qquad
m_j(t):=|\A^\midd_j(t)|,\qquad
\ell_j(t):=|\A^\low_j(t)|.
\]
Define the high-low statistic
\[
Z_j(t):=
\begin{cases}
\dfrac{h_j(t)}{h_j(t)+\ell_j(t)},&h_j(t)+\ell_j(t)>0,\\[6pt]
0,&h_j(t)+\ell_j(t)=0.
\end{cases}
\]

\end{definition}

\paragraph{Why omit the middle class?}
The sampling rule permits a reliable pairwise comparison only when two weights differ by a factor of at least \(\beta^2\). Thus, \(Z_j(t)\) compares elements above \(\beta^2t\) with elements below \(t\), leaving the ambiguous interval \([t,\beta^2t)\) aside. The gap \(\beta^2\) ensures that every high element is at least as likely to be deleted as every low element.

\begin{lemma}\label{lem:high-low-prob}
Fix a completed round \(j\) and suppose \(D_{j+1}=1\), so the adversary chooses \(\pi_j\) and samples an element \(A^*\sim\pi_j\) to delete. For every \(t\geq0\), \(A_1\in\A^\low_j(t)\), and \(A_2\in\A^\high_j(t)\),
\[
\pi_j(A_1)\leq\pi_j(A_2).
\]
Moreover, define
\[
\theta_j(t):=
\Pr\!\left(
A^*\in\A_j^\high(t)
\;\middle|\;
\Ev_j,D_{j+1}=1,A^*\notin\A_j^\midd(t)
\right).
\]
Then
\[
\theta_j(t)\geq\frac{h_j(t)}{h_j(t)+\ell_j(t)}.
\]
When the event conditioned on has probability $0$, set \(\theta_j(t):=0\).
\end{lemma}

\begin{proof}
For \(A_1\in\A_j^\low(t)\) and \(A_2\in\A_j^\high(t)\),
\[
\pi_j(A_2)
\geq\frac{1}{\beta}\frac{w_j(A_2)}{w_j(\S_j)}
\geq\beta\frac{t}{w_j(\S_j)}
\geq\beta\frac{w_j(A_1)}{w_j(\S_j)}
\geq\pi_j(A_1).
\]

The desired bound is immediate if \(h_j(t)=0\) or \(\ell_j(t)=0\), so assume both are positive. Averaging the pairwise comparison gives
\[
\frac{\sum_{A\in\A^\low_j(t)}\pi_j(A)}{\ell_j(t)}
\leq
\frac{\sum_{A\in\A^\high_j(t)}\pi_j(A)}{h_j(t)}.
\]
Consequently,
\begin{align*}
\theta_j(t)
&=
\frac{\sum_{A\in\A_j^\high(t)}\pi_j(A)}
{\sum_{A\in\A_j^\high(t)}\pi_j(A)
 +\sum_{A\in\A_j^\low(t)}\pi_j(A)}\\
&\geq
\frac{h_j(t)}{h_j(t)+\ell_j(t)}.
\end{align*}
When \(h_j(t)=\ell_j(t)=0\), both sides are zero by definition.
\end{proof}

\subsection{A supermartingale for the high-low statistic}

We next show that \((Z_j(t))_{j\geq0}\) is a bounded supermartingale for every \(t\geq0\).

\begin{lemma}\label{lem:high-low-super-martingale}

In every \(\beta\)-sampling game, for every \(j\geq0\) and \(t\geq0\),
\[
\E[Z_{j+1}(t)\mid\Ev_j]\leq Z_j(t).
\]
Consequently, for every \(0\leq s<k\),
\[
\E[Z_k(t)\mid\Ev_s]\leq Z_s(t).
\]
\end{lemma}

\begin{proof}
    We consider \(D_{j+1}=0\) and \(D_{j+1}=1\) separately. First suppose \(D_{j+1}=0\). No element is deleted, and every weight can only decrease. Hence, only the following transitions are possible:
    \[
    \high\to\midd,\qquad
    \high\to\low,\qquad
    \midd\to\low.
    \]
    Let
    \begin{align*}
        a_j(t) &:= |\A^\high_j(t) \cap \A_{j+1}^\midd(t)| \\
        b_j(t) &:= |\A^\high_j(t) \cap \A_{j+1}^\low(t)| \\
        c_j(t) &:= |\A^\midd_j(t) \cap \A_{j+1}^\low(t)| 
    \end{align*}

    and, for brevity, write \(a,b,c,h,\ell\) for their values at round \(j\) and threshold \(t\). Then
    \[
    h_{j+1}(t)=h-a-b,\qquad
    \ell_{j+1}(t)=\ell+b+c,
    \]
    so \(h_{j+1}(t)+\ell_{j+1}(t)=h+\ell+c-a\).

    If \(h+\ell=0\), then no element can enter the high class, so \(Z_j(t)=Z_{j+1}(t)=0\). If \(h+\ell>0\) but \(h+\ell+c-a=0\), then \(Z_{j+1}(t)=0\leq Z_j(t)\). Otherwise,
    \[
    Z_j(t)-Z_{j+1}(t)
    =
    \frac{a\ell+b(h+\ell)+ch}
    {(h+\ell)(h+\ell+c-a)}
    \geq0.
    \]
    Therefore,
    \[
    \E[Z_{j+1}(t)\mid\Ev_j,D_{j+1}=0]\leq Z_j(t).
    \]

    Now suppose \(D_{j+1}=1\). Let \(\widetilde Z_{j+1}(t)\) denote the statistic immediately after \(A^*\) is deleted and before the surviving weights decrease. The preceding argument shows that
    \[
    Z_{j+1}(t)\leq\widetilde Z_{j+1}(t).
    \]
    If \(A^*\in\A_j^\midd(t)\), then deletion changes neither \(h\) nor \(\ell\), and hence \(\widetilde Z_{j+1}(t)=Z_j(t)\).

    It remains to consider \(A^*\notin\A_j^\midd(t)\). If \(h=0\), then no element can enter the high class, so \(Z_{j+1}(t)=Z_j(t)=0\). If \(\ell=0\), then \(Z_j(t)=1\), and the desired inequality is immediate. We may therefore assume that \(h,\ell>0\), so \(h+\ell-1\geq1\). With \(\theta:=\theta_j(t)\), Lemma~\ref{lem:high-low-prob} gives
    \begin{align*}
    &\E\!\left[
      \widetilde Z_{j+1}(t)
      \,\middle|\,
      \Ev_j,D_{j+1}=1,A^*\notin\A_j^\midd(t)
    \right]\\
    &\qquad=
    \frac{h-1}{h+\ell-1}\theta
    +\frac{h}{h+\ell-1}(1-\theta)
    =
    \frac{h-\theta}{h+\ell-1}
    \leq
    \frac{h}{h+\ell}
    =
    Z_j(t).
    \end{align*}

    Writing
    \[
    p_{\mathrm{mid}}
    :=
    \Pr(A^*\in\A_j^\midd(t)\mid\Ev_j,D_{j+1}=1),
    \]
    and recombining the two subcases,
    \begin{align*}
    \E[\widetilde Z_{j+1}(t)\mid\Ev_j,D_{j+1}=1]
    &=
    p_{\mathrm{mid}}Z_j(t)\\
    &\quad+
    (1-p_{\mathrm{mid}})
    \E[\widetilde Z_{j+1}(t)
      \mid\Ev_j,D_{j+1}=1,A^*\notin\A_j^\midd(t)]\\
    &\leq Z_j(t).
    \end{align*}
    Since \(Z_{j+1}(t)\leq\widetilde Z_{j+1}(t)\), the same bound holds for \(Z_{j+1}(t)\). Recombining the cases \(D_{j+1}=0\) and \(D_{j+1}=1\) proves the one-step inequality. Since \(0\leq Z_j(t)\leq1\), the process is a bounded supermartingale, and the final claim follows from Lemma~\ref{lem:tower}.
\end{proof}

\subsection{Tail integration}

For the final part of the proof, define the normalized tail function of the weights on \(\S_j\):
\[
N_j(t):=
\begin{cases}
\dfrac{|\{A\in\S_j:w_j(A)\geq t\}|}{|\S_j|},&|\S_j|>0,\\[6pt]
0,&|\S_j|=0.
\end{cases}
\]

The average weight is related to \(N_j(t)\) by the layer-cake identity.

\begin{lemma}\label{lem:tail}
\[
\rho_j=\int_0^\infty N_j(t)\,dt.
\]
\end{lemma}

\begin{proof}
Recall that $\rho_j$ denotes the average weight over $\EuScript{S}_j$, i.e.
\[
\rho_j = \frac{1}{|\EuScript{S}_j|}\sum_{A \in \EuScript{S}_j} w_j(A),
\]
with the convention $\rho_j = 0$ when $|\EuScript{S}_j| = 0$ (matching the convention used to define $N_j(t)$).

\textbf{Case $|\EuScript{S}_j| = 0$.} Both sides are $0$ by convention, so the identity holds trivially.

\textbf{Case $|\EuScript{S}_j| > 0$.} We use the standard tail-integral (layer-cake) identity: for any $x \geq 0$,
\[
x = \int_0^\infty \mathbf{1}[x \geq t] \, dt.
\]
Applying this to each $w_j(A)$ for $A \in \EuScript{S}_j$ and summing,
\[
\sum_{A \in \EuScript{S}_j} w_j(A) = \sum_{A \in \EuScript{S}_j} \int_0^\infty \mathbf{1}[w_j(A) \geq t] \, dt.
\]
Since $\EuScript{S}_j$ is finite, we may exchange the (finite) sum and the integral:
\[
\sum_{A \in \EuScript{S}_j} \int_0^\infty \mathbf{1}[w_j(A) \geq t] \, dt = \int_0^\infty \sum_{A \in \EuScript{S}_j} \mathbf{1}[w_j(A) \geq t] \, dt = \int_0^\infty |\{A \in \EuScript{S}_j : w_j(A) \geq t\}| \, dt.
\]
Dividing both sides by \(|\EuScript{S}_j|\) and pulling the constant inside the integral,
\[
\rho_j = \frac{1}{|\EuScript{S}_j|}\sum_{A \in \EuScript{S}_j} w_j(A) = \int_0^\infty \frac{|\{A \in \EuScript{S}_j : w_j(A) \geq t\}|}{|\EuScript{S}_j|} \, dt = \int_0^\infty N_j(t)\, dt,
\]
using the definition of $N_j(t)$ in the last step. This completes the proof.
\end{proof}

Let us relate this to the high-low statistic defined earlier. Fix \(0\leq s<k\). We have\footnote{When the denominators are zero both sides are zero by definition.}

\[
N_s(t)
=\frac{m_s(t)+h_s(t)}{|\S_s|}
=\frac{h_s(t)+m_s(t)}{h_s(t)+m_s(t)+\ell_s(t)}
\geq
\frac{h_s(t)}{h_s(t)+\ell_s(t)}
=Z_s(t).
\]

Similarly, 

\[
N_k(\beta^2t)
=\frac{h_k(t)}{|\S_k|}
=\frac{h_k(t)}{h_k(t)+m_k(t)+\ell_k(t)}
\leq
\frac{h_k(t)}{h_k(t)+\ell_k(t)}
=Z_k(t).
\]

Using Lemma~\ref{lem:high-low-super-martingale} we have $\E[Z_k(t)|\Ev_s] \leq Z_s(t)$ so that 

\[
\E[N_k(\beta^2t)\mid\Ev_s]
\leq
\E[Z_k(t)\mid\Ev_s]
\leq
Z_s(t)
\leq
N_s(t).
\]

Integrating both sides over $t \in [0,\infty)$ gives
\[
\int_0^\infty \mathbf{E}[N_k(\beta^2 t) \mid \EuScript{E}_s] \, dt \;\leq\; \int_0^\infty N_s(t)\, dt.
\]
By Lemma~\ref{lem:tail}, the right-hand side is exactly $\rho_s$. For the left-hand side, since $N_k(\beta^2 t) \geq 0$, we may swap the (conditional) expectation and the integral over $t$ by Tonelli's theorem, and then substitute $x = \beta^2 t$ (so $dt = dx/\beta^2$):
\[
\int_0^\infty \mathbf{E}[N_k(\beta^2 t) \mid \EuScript{E}_s] \, dt
= \mathbf{E}\left[ \int_0^\infty N_k(\beta^2 t)\, dt \;\middle|\; \EuScript{E}_s \right]
= \mathbf{E}\left[ \frac{1}{\beta^2}\int_0^\infty N_k(x)\, dx \;\middle|\; \EuScript{E}_s \right]
= \frac{1}{\beta^2}\, \mathbf{E}[\rho_k \mid \EuScript{E}_s],
\]
where the last equality uses Lemma~\ref{lem:tail} applied to round $k$. Combining the two displays,
\[
\frac{1}{\beta^2}\, \mathbf{E}[\rho_k \mid \EuScript{E}_s] \;\leq\; \rho_s,
\]
and multiplying both sides by $\beta^2$ gives
\[
\mathbf{E}[\rho_k \mid \EuScript{E}_s] \;\leq\; \beta^2 \rho_s,
\]
for every \(0\leq s<k\), as required.

\section{Potential Analysis for Noisy \(k\)-Means++} \label{sec:potential-analysis}

We recall the potential-based analysis of \cite{dasgupta2013} and \cite{behs2020}. For \(0\leq\eps<1\), write
\[
r:=\frac{1+\eps}{1-\eps}.
\]

Let \(A_1,\dots,A_k\) be optimal clusters of \(X\), so \(X=\bigcup_{\nu\in[k]}A_\nu\). Up to an arbitrary assignment of boundary points, this partition corresponds to optimal centers \(C^*=\{c_1^*,\dots,c_k^*\}\), where \(c_\nu^*=\mu(A_\nu)\) and \(\mu(A):=|A|^{-1}\sum_{a\in A}a\). Thus,
\[
\opt_k(X)
=
\sum_{\nu\in[k]}\opt_1(A_\nu)
=
\sum_{\nu\in[k]}\phi(A_\nu,\mu(A_\nu)).
\]

For every set of centers \(C\), the same partition gives
\[
\phi(X,C)=\sum_{\nu\in[k]}\phi(A_\nu,C).
\]

\alg proceeds in \(k\) rounds. Let \(Y_j\) be the center sampled in round \(j\), and let \(C_j=\{Y_1,\dots,Y_j\}\). Once some sampled center belongs to an optimal cluster, that cluster's expected contribution is bounded by an \(O_\eps(1)\) factor of its optimum contribution. We use the following standard lemmas.
\begin{lemma}  \label{lem:first-center}(\citet[Lemma 4]{behs2020}, adapted from \citet[Lemma 5]{dasgupta2013})
    Let $C_1 = \{Y_1\}$ denote the first center chosen by \alg. Then for any optimal cluster $A_\nu$, 
    $$ \E[\phi(A_\nu,C_1) | Y_1 \in A_\nu] \leq 2r \cdot \opt_1(A_\nu) $$
\end{lemma}

\begin{lemma}\label{lem:other-centers} (\citet[Lemma 5]{behs2020}, adapted from \citet[Lemma 6]{dasgupta2013})
    Consider an iteration $1 \leq j <k$ of \alg after the first one and suppose $C_j \neq \emptyset$ are the centers chosen until now. Let the center chosen in the considered iteration be $Y_{j+1}$. Then for any optimal cluster $A_\nu$ 
    $$ \E[\phi(A_\nu,C_{j+1}) | \F_j, Y_{j+1} \in A_\nu] \leq 8r \cdot \opt_1(A_\nu) $$
\end{lemma}

This does not yield an \(O_\eps(1)\) guarantee because several sampled centers may belong to the same optimal cluster. We therefore distinguish covered and uncovered optimal clusters.

\begin{definition}\label{def:cover-uncover}
Let $1 \leq j \leq k$.  An optimal cluster $A_\nu$ is said to be covered after round $j$ if at least one of the first $j$ sampled centers belongs to it. We have the following definitions: 

$$\H_j = \{A_\nu : C_j \cap A_\nu \neq \emptyset\} \quad \U_j = \{A_\nu : C_j \cap A_\nu = \emptyset\} \quad u_j = |\U_j|$$ 

$$ H_j = \bigcup_{A_\nu \in \H_j} A_\nu \quad U_j = \bigcup_{A_\nu \in \U_j} A_\nu $$

After round \(j\), the dataset is partitioned as \(X=H_j\cup U_j\). Since at most \(j\) optimal clusters can be covered, \(u_j\geq k-j\). Define
\[
I_j:=\one\{Y_j\in H_{j-1}\},
\]
so \(I_j=1\) exactly when round \(j\) is wasted.

\end{definition}

We use \(H_j,U_j\) for sets of points and the calligraphic notation \(\H_j,\U_j\) for sets of clusters. We also record the sampling history.
\begin{definition} \label{def:filtration}
    Let \(Y_j\in X\) denote the point sampled in round \(j\), with the
initial sample regarded as round \(1\), and let
\[
    C_j:=\{Y_1,\ldots,Y_j\}.
\]
We write
\[
    \F_0:=\{\varnothing,\Omega\},
    \qquad
    \F_j:=\sigma(Y_1,\ldots,Y_j),
    \quad j=1,\ldots,k,
\]
for the natural filtration generated by the sampling process. Here, $\Omega$ is the set of all possible executions of the algorithm.  Thus,
\(\F_j\) represents the complete history of the algorithm up
to and including round \(j\). In particular, the random variables
\(C_j,H_j,U_j,u_j\), and \(I_j\) are
\(\F_j\)-measurable.
\end{definition}

Lemmas~\ref{lem:first-center} and~\ref{lem:other-centers} imply the following bound.

\begin{lemma}\label{lem:covered-clusters} (\citet[Corollary 6]{behs2020}) For each $j$ we have 
$$\E[\phi(H_j,C_j)] \leq 8r \cdot \opt_k(X)$$
\end{lemma}

The main technical challenge is to bound the uncovered clusters. Following \cite{behs2020}, define
\[
\rho_j :=
\begin{cases}
\dfrac{\phi(U_j,C_j)}{u_j}, & u_j>0,\\[1ex]
0, & u_j=0,
\end{cases}
\qquad
\Psi_k := \sum_{j=2}^k I_j \rho_j,
\] 

We use the following bound from \citet[Lemma~7]{behs2020}.

\begin{lemma}\label{lem:upper-bound-potential}
$$
\E[\Psi_k]
\leq
8r^2(\ln k+1)\opt_k(X).
$$
\end{lemma}

\paragraph{Proof sketch.}
Fix \(2\leq j\leq k\) and condition on \(\F_{j-1}\). If round \(j\) is wasted, then \(\U_j=\U_{j-1}\), \(u_j=u_{j-1}\), and adding a center cannot increase the cost of any uncovered cluster. Hence,
\[
\E[I_j\rho_j\mid\F_{j-1}]
\leq
\Pr(I_j=1\mid\F_{j-1})\rho_{j-1}.
\]
If \(\phi(U_{j-1},C_{j-1})=0\), this conditional contribution is zero and there is nothing to prove. Otherwise, the noisy sampling inequalities imply
\[
\Pr(I_j=1\mid\F_{j-1})
\leq
r\,\frac{\phi(H_{j-1},C_{j-1})}
        {\phi(U_{j-1},C_{j-1})}.
\]
Indeed, the numerator is at most \(1+\eps\) times its exact \(D^2\)-mass, the uncovered mass is at least \(1-\eps\) times its exact mass, and the probability of a wasted round is at most the ratio of these two noisy masses. Since
\(\rho_{j-1}=\phi(U_{j-1},C_{j-1})/u_{j-1}\) and \(u_{j-1}\geq k-j+1\),
\[
\E[I_j\rho_j]
\leq
\frac{r}{k-j+1}
\E[\phi(H_{j-1},C_{j-1})]
\leq
\frac{8r^2}{k-j+1}\opt_k(X),
\]
where the last inequality is Lemma~\ref{lem:covered-clusters}. Summing over \(j=2,\dots,k\) and using
\(\sum_{h=1}^{k-1}h^{-1}\leq\ln k+1\) proves the claim.

The following consequence of Theorem~\ref{thm:improved-sampling} is the main new ingredient in our analysis.

\begin{lemma}\label{lem:future-uncovered}
Fix a round $2\leq s\leq k$ and condition on the complete history $\F_s$. Then
$$
\E[\rho_k\mid\F_s]
\leq
r^2\rho_s,
\qquad
r=\frac{1+\eps}{1-\eps}.
$$
\end{lemma}

\begin{proof}
We show that, after conditioning on \(\F_s\), the evolution of the uncovered optimal clusters is an instance of the \(\beta\)-sampling game with \(\beta=r\).

For $s\leq j\leq k$ and $A\in\U_j$, define
$$
w_j(A):=\phi(A,C_j).
$$
The total weight of the uncovered clusters is
$$
w_j(\U_j)
:=
\sum_{A\in\U_j}w_j(A)
=
\phi(U_j,C_j).
$$
Consequently,
$$
\rho_j=
\begin{cases}
\dfrac{w_j(\U_j)}{u_j},&u_j>0,\\[6pt]
0,&u_j=0,
\end{cases}
$$
is precisely their average weight.

If $u_s=0$, then all optimal clusters are already covered and both sides of the desired inequality are $0$. Similarly, if $w_s(\U_s)=0$, then the weights remain $0$ in every subsequent round because centers are only added. We may therefore assume that
$$
u_s>0
\qquad\text{and}\qquad
w_s(\U_s)>0.
$$

Fix $s\leq j<k$ and condition on $\F_j$. Let
$$
q_{j+1}(x)
:=
\Pr(Y_{j+1}=x\mid\F_j),
\qquad x\in X,
$$
be the noisy distribution used to sample the next center. By the noisy $D^2$-sampling rule,
$$
(1-\eps)
\frac{\phi(x,C_j)}{\phi(X,C_j)}
\leq
q_{j+1}(x)
\leq
(1+\eps)
\frac{\phi(x,C_j)}{\phi(X,C_j)}
$$
for every $x\in X$.

Suppose that
$$
\Pr(I_{j+1}=0\mid\F_j)>0.
$$
For each $A\in\U_j$, define
$$
\pi_j(A)
:=
\Pr(Y_{j+1}\in A\mid\F_j,I_{j+1}=0).
$$
Thus, $\pi_j(A)$ is the probability that $A$ is the newly covered cluster, conditioned on the round being non-wasted. Summing the noisy $D^2$-sampling inequalities over all points in $A$ gives
$$
(1-\eps)
\frac{w_j(A)}{\phi(X,C_j)}
\leq
\Pr(Y_{j+1}\in A\mid\F_j)
\leq
(1+\eps)
\frac{w_j(A)}{\phi(X,C_j)}.
$$
Moreover, the event $I_{j+1}=0$ is exactly the event that the sampled center belongs to one of the uncovered clusters. Therefore, summing over $A\in\U_j$ gives
$$
(1-\eps)
\frac{w_j(\U_j)}{\phi(X,C_j)}
\leq
\Pr(I_{j+1}=0\mid\F_j)
\leq
(1+\eps)
\frac{w_j(\U_j)}{\phi(X,C_j)}.
$$
Since the conditioning event has positive probability, we may divide the two bounds. Hence, for every $A\in\U_j$,
$$
\begin{aligned}
\pi_j(A)
&=
\frac{\Pr(Y_{j+1}\in A\mid\F_j)}
     {\Pr(I_{j+1}=0\mid\F_j)}\\
&\geq
\frac{1-\eps}{1+\eps}
\frac{w_j(A)}{w_j(\U_j)}
=
\frac{1}{r}\frac{w_j(A)}{w_j(\U_j)}.
\end{aligned}
$$
Similarly,
$$
\pi_j(A)
\leq
\frac{1+\eps}{1-\eps}
\frac{w_j(A)}{w_j(\U_j)}
=
r\frac{w_j(A)}{w_j(\U_j)}.
$$
Thus,
\begin{equation}
\label{eq:conditional-cluster-noise}
\frac{1}{r}\frac{w_j(A)}{w_j(\U_j)}
\leq
\pi_j(A)
\leq
r\frac{w_j(A)}{w_j(\U_j)}.
\end{equation}

We now verify the remaining rules of the sampling game. After conditioning on \(\F_s\), take \(\U_s\) as the initial set and assign each \(A\in\U_s\) its actual weight \(w_s(A)=\phi(A,C_s)\). We identify round \(\tau\) of the sampling game with round \(s+\tau\) of noisy \(D^2\)-sampling. Under this identification, the game's deletion event \(D_{\tau+1}=1\) corresponds exactly to the non-wasted event \(I_{s+\tau+1}=0\).

Consider the transition from round $j$ to round $j+1$. There are two possibilities.

\begin{enumerate}
    \item If $I_{j+1}=1$, then the round is wasted: $Y_{j+1}\in H_j$, and no new optimal cluster is covered. Hence,
    $$
    \U_{j+1}=\U_j.
    $$
    Since $C_j\subseteq C_{j+1}$, the weight of every uncovered cluster can only decrease:
    $$
    0\leq w_{j+1}(A)
    =
    \phi(A,C_{j+1})
    \leq
    \phi(A,C_j)
    =
    w_j(A)
    $$
    for every $A\in\U_j$.

    \item If $I_{j+1}=0$, then $Y_{j+1}\in U_j$. Because the optimal clusters form a partition, there is a unique cluster $A^\star\in\U_j$ containing $Y_{j+1}$. This cluster becomes covered and is removed:
    $$
    \U_{j+1}
    =
    \U_j\setminus\{A^\star\}.
    $$
    Conditioned on \(I_{j+1}=0\), the cluster \(A^\star\) is chosen according to \(\pi_j\), which satisfies \eqref{eq:conditional-cluster-noise}. Furthermore, the weights of all surviving clusters can only decrease:
    $$
    0\leq w_{j+1}(A)\leq w_j(A)
    $$
    for every $A\in\U_{j+1}$.
\end{enumerate}

If \(\Pr(I_{j+1}=0\mid\F_j)=0\), only the first possibility occurs, which is also allowed by the sampling game. Thus, conditional on \(\F_s\), the entire evolution from round \(s\) to round \(k\) is an instance of the \(\beta\)-sampling game with \(\beta=r\).

Applying Theorem~\ref{thm:improved-sampling} to this conditional process gives
$$
\E[\rho_k\mid\F_s]
\leq
\beta^2\rho_s
=
r^2\rho_s,
$$
as claimed.
\end{proof}

\paragraph{Discussion.}
The lemma shows that the expected average cost of uncovered clusters increases by at most \(B(k,\eps)=r^2\). \cite{behs2020} proved the weaker bound \(B_1(k,\eps)=4r\ln k+2\) for sufficiently large \(k\), and \cite{gor2023} later obtained \(B_2(k,\eps)\leq\frac{90}{(1-e^{-1/40})^2}\approx147{,}638\). Our dynamic-threshold argument yields \(r^2=1+O(\eps)\).

The potential bound and Lemma~\ref{lem:future-uncovered} combine through the following standard consequence of the potential analysis.

\begin{lemma}\label{lem:uncovered-upper-bound}(\citet[Lemma 8]{behs2020})
For \(B(k,\eps)=r^2\) as in Lemma~\ref{lem:future-uncovered},
\[
\E[\phi(U_k,C_k)]\leq B(k,\eps)\E[\Psi_k].
\]
\end{lemma}

\paragraph{Proof of Theorem~\ref{thm:main}.}
Because \(H_k\) and \(U_k\) partition \(X\),
\[
  \phi(X,C_k)
  =
  \phi(H_k,C_k)+\phi(U_k,C_k).
\]
Applying Lemmas~\ref{lem:covered-clusters},~\ref{lem:uncovered-upper-bound}, and~\ref{lem:upper-bound-potential} gives
\begin{align*}
  \E[\phi(X,C_k)]
  &\le
  8r\cdot \opt_k(X)+B(k,\eps) \cdot \E[\Psi_k]\\
  &\le
  8r \cdot \opt_k(X)
  +8r^4 \cdot ( \ln k+1)\cdot \opt_k(X)\\
  &\leq
  8 (\ln k+2)\left( \frac{1+\eps}{1-\eps} \right)^4 \opt_k(X).
\end{align*}

Finally,
\[
\left(\frac{1+\eps}{1-\eps}\right)^4
\leq1+160\eps
\qquad\text{for }0\leq\eps\leq\frac12,
\]
which follows, for example, by convexity and comparison with the chord joining the endpoints. This completes the proof.

\section{Lower Bounds}
\label{sec:lower-bounds}

We complement the upper bound with two lower bounds.  The first shows
that a linear loss in the multiplicative-noise parameter is unavoidable.
The second shows that replacing pointwise multiplicative control by total
variation closeness destroys every instance-independent approximation
guarantee, already for \(k=2\).

\subsection{Linear dependence on $\eps$ is necessary}

Let $\mathcal{H}_k = \sum_{m=1}^k1/m$. We consider the lower bound family of well-separated high-dimensional simplices of \cite{av2007}\footnote{See also the simplified analysis of this lower bound instance by \cite{aggarwal2009}. } and construct an adversary that multiplies the probability of a round being wasted by a factor of $(1+\eps)$. We denote expectation under this adversary by $\E_\eps[\cdot ]$ and for exact $k$-means++ by $\E_0[\cdot]$. Our lower bound result is as follows: 

\begin{theorem} \label{thm:linear-noise-lower-bound}
For every $k \geq 2$ and $0 \leq \eps <1$, there exists a family $X := X_{\delta,\Delta}$ of Euclidean $k$-means instances and a valid $\eps$-noisy adversary such that: 
$$ \lim_{\Delta / \delta \to \infty} \frac{\E_\eps[\phi(X,C)]}{\E_0[\phi(X,C)]} \ge 1 + \frac{\mathcal{H}_k-1}{\mathcal{H}_k} \eps \geq 1 + \frac{\eps}{3}$$
Consequently, a multiplicative loss of $1+ \Omega(\eps)$ with respect to exact $k$-means++ is unavoidable.  
\end{theorem}

    \emph{Construction.} Fix the number of clusters $k \geq 2$. Let the number of points be $n := km$ so that each cluster $A_1,\dots,A_k$ has exactly $m = n/k$ points in a simplex of side length $\delta$. Place the simplices in mutually orthogonal subspaces and choose their centroids $c_1,\dots,c_k$ separated by pairwise squared distance $\Delta^2 - \frac{m-1}{m}\delta^2$. For sufficiently large $\Delta/\delta$, the optimal clusters are indeed $A_1,\dots,A_k$. Moreover, this construction satisfies the following property: 

    $$ \| x-y\|^2 = \begin{cases}
        \delta^2 &: x,y \in A_i, x\neq y \\ \Delta^2&: x \in A_i, y \in A_j, i \neq j
    \end{cases} $$

    \emph{Adversary.} The adversary leaves the first sample (which is sampled uniformly at random) unchanged. Suppose the set of centers sampled until now is $C$. Recall that $C$ will partition the optimal clusters into covered or hit clusters $H$ and uncovered clusters $U$ such that $X = H \cup U$. Suppose $p$ is the exact $D^2$ distribution with respect to $C$. Let $p(H) := \sum_{x \in H} p(x)$. The adversary chooses the noisy distribution $q$ as follows: 

    $$ q(x) := \begin{cases}
        (1+\eps) \cdot p(x) &: x \in H \\\\ 
        \left(1 - \eps \cdot \frac{p(H)}{1 - p(H)}\right) \cdot p(x) &: x \in U
        
    \end{cases} $$

Let us see under which conditions the above distribution is valid. Since there is at least one uncovered cluster before the last center is chosen, we have $p(H) < 1$. Since $H \subset X$ and each point $x \in H$ can contribute at most $\delta^2$ to the cost, we have $\phi(H,C) \leq n \delta^2$. Similarly since there is at least one uncovered cluster, each point $x \in U$ contributes $\Delta^2$ cost, we have $\phi(U,C) \geq m\Delta^2$. Hence 
$$p(H)= \frac{\phi(H,C)}{\phi(X,C)} \leq \frac{\phi(H,C)}{\phi(U,C)} \leq \frac{n\delta^2}{m\Delta^2} = k \cdot \frac{\delta^2}{\Delta^2} \leq \frac{1}{2}$$

whenever $\Delta^2/\delta^2 \geq 2k$. This implies $\frac{p(H)}{1 - p(H)} \leq 1$ so that $\left( 1- \eps \cdot \frac{p(H)}{1-p(H)}\right) \geq 1 - \eps$. So indeed for each $x \in X$ we have $q(x) \in (1 \pm \eps)p(x)$. We assume that this condition holds for the rest of the proof.\\

To analyse the behaviour of $D^2$-sampling, it is crucial to see that the \emph{state} of the algorithm can be adequately described by two parameters: the total number of centers sampled by the algorithm until now: $t$ and the total number of optimal clusters \emph{hit} until now: $h$. 
\begin{lemma}\label{lem:cost}
Suppose that during the execution of Algorithm~\ref{alg:noisy} on the instance described above, the current set of centers $C$ has $t$ centers and $h$ out of the $k$ optimal clusters have been hit until now. Then, the $k$-means cost of $X$ with respect to $C$ is given by: 
$$ \phi(X,C)  = \Phi[h,t] := (mh-t)\delta^2 + (k-h)m\Delta^2$$
\end{lemma}

\begin{proof}
    Recall that the optimal clusters are $A_1,\dots,A_k$. Suppose that the set of centers $C$ has hit $h$ optimal clusters $A_{i_1},\dots,A_{i_h}$. Also suppose that each hit cluster $i_\ell$ has exactly $s_\ell$ centers from $C$ belonging to it i.e., $|C \cap A_{i_\ell}| = s_\ell$ so that $\sum_{\ell=1}^h s_\ell = t$. Now consider $A_{i_\ell}$, which has $m$ points, out of which $m - s_\ell$ contribute $\delta^2$ cost while $s_\ell$ of them contribute zero. Hence we have $\phi(A_{i_\ell},C) = (m-s_\ell)\delta^2$ so that $\phi(H,C) = \sum_{\ell = 1}^h \phi(A_{i_\ell},C) = \sum_{\ell = 1}^h (m-s_\ell)\delta^2 = (mh-t)\delta^2$. Now consider the uncovered clusters $U$: each point $x \in U$ contributes $\Delta^2$ to the cost. Hence $\phi(U,C) = (k-h)\cdot m\cdot \Delta^2$, This finally gives $\Phi[h,t] = \phi(X,C) = \phi(H,C)+\phi(U,C) = (mh-t)\delta^2 + (k-h)m\Delta^2$. 
\end{proof}

To compute the expected cost of $D^2$-sampling on this instance, we introduce the \emph{dynamic programming table} $F$ such that $F_\sigma[h,t]$ is the expected $k$-means cost, starting from a set of centers having $t$ total centers which have \emph{hit} exactly $h$ optimal clusters, and sampling the remaining $k-t$ centers; $\sigma = 0$ corresponds to using exact $D^2$ sampling  and $\sigma = \eps$ corresponds to sampling according to the $\eps$-noisy adversarial distribution which we had defined earlier. Observe that the first center chosen corresponds to $h=1,t=1$ and hence we can write $\E_\sigma[\phi(X,C)] = F_\sigma[1,1]$ where $C$ is the set of centers output by Algorithm~\ref{alg:noisy}. We get the following dynamic programming recurrence for $F_\sigma[h,t]$: 

\begin{lemma}\label{lem:dp-recurrence}
For $\sigma \in \{0,\eps\}$ let $F_\sigma[h,t]$ be the expected final cost conditioned on having already chosen $t$ centers all of which lie in $h$ distinct optimal clusters. For every state $(h,t)$ with $t < k$, we have: 
$$  F_\sigma[h,t] = (1+\sigma) \cdot w[h,t] \cdot F_\sigma[h,t+1] + \left(1 - (1+\sigma) \cdot w[h,t]\right) \cdot F_\sigma[h+1,t+1]$$
where $$w[h,t] := \frac{(mh-t)\delta^2}{(mh-t)\delta^2 + (k-h)m\Delta^2}$$
with the terminal condition
$$F_\sigma[h,k] = (hm-k)\delta^2 + (k-h)m\Delta^2 $$
\end{lemma}

\begin{proof}
Fix a state $(h,t)$ with $t<k$, and let $H$ and $U$ denote the unions of the hit and unhit optimal clusters, respectively. By Lemma~\ref{lem:cost},
\[
\phi(H,C)=(mh-t)\delta^2
\qquad\text{and}\qquad
\phi(U,C)=(k-h)m\Delta^2.
\]
Hence, under exact $D^2$-sampling, the probability that the next center lies in a hit cluster is
\[
p(H)
=
\frac{\phi(H,C)}{\phi(X,C)}
=
\frac{(mh-t)\delta^2}
     {(mh-t)\delta^2+(k-h)m\Delta^2}
=
w[h,t].
\]
For $\sigma=0$, this probability is $w[h,t]$. For $\sigma=\eps$, the adversary multiplies it by $1+\eps$. Thus, in either case, the probability of sampling from a hit cluster is
\[
(1+\sigma)w[h,t],
\]
and the probability of sampling from an unhit cluster is
\[
1-(1+\sigma)w[h,t].
\]

If the next center lies in $H$, then the number of hit clusters remains $h$, while the number of sampled centers increases to $t+1$. The resulting state is therefore $(h,t+1)$. If the next center lies in $U$, then exactly one previously unhit optimal cluster is hit, and the resulting state is $(h+1,t+1)$. Conditioning on these two events yields
\[
F_\sigma[h,t]
=
(1+\sigma)w[h,t]F_\sigma[h,t+1]
+
\left(1-(1+\sigma)w[h,t]\right)
F_\sigma[h+1,t+1].
\]

Finally, when $t=k$, no further centers are sampled. Lemma~\ref{lem:cost} therefore gives
\[
F_\sigma[h,k]
=
\Phi[h,k]
=
(hm-k)\delta^2+(k-h)m\Delta^2,
\]
which proves the recurrence and the terminal condition.
\end{proof}

To prove Theorem~\ref{thm:linear-noise-lower-bound}, we require a lower bound on $F_\eps[1,1] = \E_\eps[\phi(X,C)]$ and an upper bound on $F_0[1,1] = \E_0[\phi(X,C)]$. For obtaining this, we first introduce the random variable $I$, which is the number of \emph{wasted} rounds in an execution of Algorithm~\ref{alg:noisy}, so that at the end of the algorithm, there are exactly $I$ uncovered clusters and $k-I$ covered clusters. Using Lemma~\ref{lem:cost} we can directly write $\phi(X,C) = (m(k-I)-k)\delta^2 + Im\Delta^2 = k(m-1)\delta^2 + mI(\Delta^2-\delta^2)$ by substituting $t = k $ and $h = k-I$. Hence, the quantity which we want a bound on can be written as: 

\begin{lemma}\label{lem:cost-waste-relation}
    $$ F_\sigma[1,1] = k(m-1)\delta^2 + m(\Delta^2-\delta^2) \cdot \E_\sigma[I] $$
\end{lemma}

So, it remains to show a lower bound on $\E_\eps[I]$ and an upper bound on $\E_0[I]$: 

\begin{lemma} \label{lem:upper-lower-bound}
Putting $B := \frac{m-1}{m} \cdot \frac{k\delta^2}{\Delta^2} \cdot \sum_{\ell=2}^k \frac{1}{\ell}$ we have the following bounds: 
\begin{enumerate}
    \item $$\E_0[I] \leq B$$
    \item $$\E_\eps[I] \geq 1 - \frac{e^{-\eps B}}{1+B}$$
\end{enumerate}
\end{lemma}
\begin{proof}
    It will help to define the random variable $I$ formally first. Define the indicator random variable $I_t$ so that $I_t = 0$ when a new cluster is hit in the $(t+1)$st round and $1$ otherwise. Hence we may write $I = \sum_{t=1}^{k-1} I_t$. Suppose that after round $t$, $h$ clusters have already been hit. We have
    $$ \E_0[I_t] = \Pr[I_t=1] = w[h,t] = \frac{(mh-t)\delta^2}{(mh-t)\delta^2 + (k-h)m\Delta^2} \leq \frac{(mh-t)\delta^2}{(k-h)m\Delta^2} \leq \frac{t}{k-t} \frac{m-1}{m} \frac{\delta^2}{\Delta^2}$$
    where we used $h\leq t$ and $t\leq k$. Now

$$\E_0[I] = \sum_{t=1}^{k-1}\E_0[I_t] \leq \frac{m-1}{m} \frac{\delta^2}{\Delta^2} \sum_{t=1}^{k-1} \frac{t}{k-t} = B$$

Now consider the lower bound for the noisy case. Markov's inequality gives $$\E_\eps[I] \geq \Pr[I \geq 1]=1-\Pr[I = 0]$$
Note that $I = 0$ corresponds to each new center being sampled being from a previously uncovered cluster. Under the adversary, this happens at round $t$ with probability $1 - (1+\eps) w[t,t]$; here $h = t$, which holds on the event $I = 0$. Now consider:

$$ 1 - (1+\eps) w[t,t] = 1-(1+\eps) \frac{b[t]}{1+b[t]} = \frac{1 - \eps b[t]}{1+ b[t]}$$

where $b[t]:= \frac{m-1}{m} \frac{t}{k-t} \frac{\delta^2}{\Delta^2}$. Multiplying these probabilities gives\footnote{We use the fact that $1-x \leq e^{-x}$ for $0 \leq x \leq 1$ and for positive numbers $x_1,\dots,x_\ell$ we have $1 + \sum_{t=1}^\ell x_t \leq \prod_{t=1}^\ell (1+x_t)$.}: 

$$\Pr[I = 0] = \prod_{t=1}^{k-1} \frac{1 - \eps b[t]}{1 + b[t]} \leq \prod_{t=1}^{k-1} \frac{e^{-\eps b[t]}}{1+b[t]} \leq \frac{e^{-\eps \sum_{t=1}^{k-1} b[t]}}{1 + \sum_{t=1}^{k-1} b[t]} = \frac{e^{-\eps B}}{1+B}$$
This completes the proof.
\end{proof}

Now we have all the ingredients for the proof of Theorem~\ref{thm:linear-noise-lower-bound}. Note that as $m \to \infty$ and $\frac{\Delta^2}{\delta^2} \to  \infty$, $B\to 0$. So we have from Lemma~\ref{lem:cost-waste-relation}: 

$$ \frac{F_\eps[1,1]}{F_0[1,1]} = \frac{k(m-1)\delta^2 + m(\Delta^2-\delta^2) \cdot \E_\eps[I]}{k(m-1)\delta^2 + m(\Delta^2-\delta^2) \cdot \E_0[I]} \geq \frac{1 + (1 - \delta^2/\Delta^2)(\mathcal H_k-1) \frac{\left(1 - \frac{e^{-\eps B}}{1+B}\right)}{B}}{1 + (1 - \delta^2/\Delta^2)(\mathcal H_k-1)}$$

Taking the limit $\Delta^2/\delta^2 \to \infty$ we obtain: 

$$ \frac{F_\eps[1,1]}{F_0[1,1]} \geq 1+\eps \frac{\mathcal{H}_k-1}{\mathcal{H}_k} \geq 1 + \frac{\eps}{3}$$

This completes the proof. 

\subsection{Total variation closeness gives no guarantee}
\label{sec:tv-lower-bound}

For distributions \(p,q\) on a finite set $X$, recall that
\[
 \|p-q\|_{\mathrm{TV}}
 :=\frac12\sum_{x\in X}|p(x)-q(x)|.
\]
Consider the relaxation of noisy \(k\)-means++ in which the
distribution \(q_j\) used in every round is only required to satisfy
\(\|q_j-p_{C_j}\|_{\mathrm{TV}}\leq\eta\).

\begin{theorem}
\label{thm:tv-no-guarantee}
For every \(0<\eta<1\) and every \(M<\infty\), there is a Euclidean
\(2\)-means instance and a support-preserving sampling rule whose
nontrivial sampling round is within total variation distance
\(\eta\) of exact \(D^2\)-sampling, but whose expected approximation
ratio is greater than \(M\).  Thus no finite approximation guarantee
depending only on \(k\) and \(\eta\) is possible under this
relaxation.
\end{theorem}

\begin{proof}
Let
\[
 A=\{a_1,a_2\},\qquad B=\{b_1,b_2\},
\]
where the squared distance within each pair is \(\delta^2\), while
every cross-pair squared distance is \(\Delta^2\). Four such points
exist in \(\mathbb R^3\): for
\(L^2=\Delta^2-\delta^2/2\), take
\[
 a_1=(-\delta/2,0,0),\quad
 a_2=(\delta/2,0,0),\quad
 b_1=(0,-\delta/2,L),\quad
 b_2=(0,\delta/2,L).
\]
The centers at the midpoints of \(A\) and \(B\) have total cost
\(\delta^2\), and hence
\[
 \operatorname{OPT}_2(X)\leq\delta^2.
\]

Choose the first center uniformly, exactly as in \(k\)-means++.
Suppose without loss of generality that it is \(a_1\).  On the three
remaining points, exact \(D^2\)-sampling uses
\[
 p(a_2)=t:=\frac{\delta^2}{\delta^2+2\Delta^2},
 \qquad
 p(b_1)=p(b_2)=\frac{1-t}{2}.
\]
Choose \(\Delta/\delta\) large enough that \(\eta<1-t\), and define
\[
 q(a_2)=t+\eta,\qquad
 q(b_1)=q(b_2)=\frac{1-t-\eta}{2}.
 \]
Define the rule symmetrically for every possible first center.  It is
a normalized, support-preserving distribution, and
\[
 \|p-q\|_{\mathrm{TV}}
 =
 \frac12\left(\eta+\frac{\eta}{2}+\frac{\eta}{2}\right)
 =\eta.
\]

With probability \(q(a_2)=t+\eta\geq\eta\), both selected centers lie
in \(A\).  On this event both points of \(B\) are at squared distance
\(\Delta^2\) from the selected centers, so the final cost is
\(2\Delta^2\).  It follows that
\[
 \frac{\E_q[\varphi(X,C)]}{\operatorname{OPT}_2(X)}
 \geq
 2\eta\frac{\Delta^2}{\delta^2}.
\]
The right-hand side exceeds \(M\) for a sufficiently large
\(\Delta/\delta\).
\end{proof}

\begin{corollary}
\label{cor:vanishing-tv}
There is a sequence of \(2\)-means instances and support-preserving
sampling distributions \(q_R\) such that
\[
 \|q_R-p_R\|_{\mathrm{TV}}\longrightarrow0,
 \qquad
 \frac{\E_{q_R}[\varphi(X_R,C)]}
      {\operatorname{OPT}_2(X_R)}
 \longrightarrow\infty.
\]
\end{corollary}

\begin{proof}
In the construction of Theorem~\ref{thm:tv-no-guarantee}, put
\(R=\Delta/\delta\) and \(\eta_R=1/R\).  Then the TV distance tends
to zero, whereas the expected approximation ratio is at least
\(2R\).
\end{proof}

\section{Conclusion and Discussion}
In this work, we studied the robustness of the $k$-means++ algorithm and resolved an open question posed by \cite{gor2023}. It would be very interesting to see if our techniques extend to noisy variants of other algorithms which use $D^2$-sampling such as \emph{greedy} $k$-means++ \citep{behs2020,gort} or $k$-means$\|$\citep{bahmani}.

\section*{Acknowledgments}

The author gratefully thanks Ragesh Jaiswal for his guidance and mentorship. The author thanks  Sandeep Silwal for discussions and comments on an earlier draft which helped improve the paper.\\

\emph{Declaration of AI use.}
The dynamic-threshold reformulation underlying this paper - replacing the fixed big/medium/small classes of \cite{gor2023} with threshold classes defined at every scale, and recovering the average uncovered cost via integration was conceived by the author. During the exploratory stage, the author used OpenAI's GPT-5.6-Sol to develop this idea into a concrete argument, including the supermartingale formulation and layer-cake integration step, and to generate a preliminary proof sketch. All model output was treated as unverified: the author independently reconstructed, verified, and wrote every lemma, proof, and line of exposition in the final manuscript. Anthropic's Claude Opus 5 was used afterward only as an editorial reviewer of the completed draft. No text from either system was copied into the manuscript, no AI tool was used for citations, and the author assumes full responsibility for all content.

\bibliographystyle{plainnat}
\bibliography{paper}

@inproceedings{av2007,
  author    = {Arthur, David and Vassilvitskii, Sergei},
  title     = {{k}-means++: The Advantages of Careful Seeding},
  booktitle = {Proceedings of the Eighteenth Annual ACM-SIAM Symposium on
               Discrete Algorithms},
  series    = {SODA '07},
  year      = {2007},
  pages     = {1027--1035},
  publisher = {Society for Industrial and Applied Mathematics},
  address   = {Philadelphia, PA, USA},
  url       = {https://theory.stanford.edu/~sergei/papers/kMeansPP-soda.pdf}
}

@inproceedings{behs2020,
  author    = {Bhattacharya, Anup and Eube, Jan and R{\"o}glin, Heiko and
               Schmidt, Melanie},
  title     = {Noisy, Greedy and Not So Greedy {k}-Means++},
  booktitle = {28th Annual European Symposium on Algorithms (ESA 2020)},
  series    = {Leibniz International Proceedings in Informatics (LIPIcs)},
  volume    = {173},
  year      = {2020},
  pages     = {18:1--18:21},
  doi       = {10.4230/LIPIcs.ESA.2020.18},
  publisher = {Schloss Dagstuhl--Leibniz-Zentrum f{\"u}r Informatik}
}

@misc{dasgupta2013,
  author       = {Dasgupta, Sanjoy},
  title        = {Lecture 3: Algorithms for {k}-Means Clustering},
  howpublished = {Course notes, University of California, San Diego},
  year         = {2013},
  url          = {https://cseweb.ucsd.edu/~dasgupta/291-geom/kmeans.pdf}
}

@inproceedings{gor2023,
  author    = {Grunau, Christoph and {\"O}z{\"u}do{\u{g}}ru, Ahmet Alper and
               Rozho{\v{n}}, V{\'a}clav},
  title     = {Noisy {k}-Means++ Revisited},
  booktitle = {31st Annual European Symposium on Algorithms (ESA 2023)},
  series    = {Leibniz International Proceedings in Informatics (LIPIcs)},
  volume    = {274},
  year      = {2023},
  pages     = {55:1--55:7},
  doi       = {10.4230/LIPIcs.ESA.2023.55},
  publisher = {Schloss Dagstuhl--Leibniz-Zentrum f{\"u}r Informatik}
}

@article{goldberg1991,
  author  = {Goldberg, David},
  title   = {What Every Computer Scientist Should Know About
             Floating-Point Arithmetic},
  journal = {ACM Computing Surveys},
  volume  = {23},
  number  = {1},
  pages   = {5--48},
  year    = {1991},
  doi     = {10.1145/103162.103163}
}

@article{pedregosa2011scikit,
  author  = {Pedregosa, Fabian and Varoquaux, Ga{\"e}l and Gramfort,
             Alexandre and Michel, Vincent and Thirion, Bertrand and
             Grisel, Olivier and Blondel, Mathieu and Prettenhofer,
             Peter and Weiss, Ron and Dubourg, Vincent and Vanderplas,
             Jake and Passos, Alexandre and Cournapeau, David and
             Brucher, Matthieu and Perrot, Matthieu and Duchesnay,
             {\'E}douard},
  title   = {Scikit-learn: Machine Learning in Python},
  journal = {Journal of Machine Learning Research},
  volume  = {12},
  number  = {85},
  pages   = {2825--2830},
  year    = {2011},
  url     = {https://www.jmlr.org/papers/v12/pedregosa11a.html}
}

@misc{sklearnKMeans,
  author       = {{The scikit-learn developers}},
  title        = {\texttt{sklearn.cluster.KMeans}},
  howpublished = {scikit-learn documentation},
  year         = {2026},
  url          = {https://scikit-learn.org/stable/modules/generated/sklearn.cluster.KMeans.html},
  note         = {Accessed July 27, 2026}
}

@incollection{williams1991martingales,
  author    = {Williams, David},
  title     = {Martingales},
  booktitle = {Probability with Martingales},
  chapter   = {10},
  pages     = {93--121},
  publisher = {Cambridge University Press},
  address   = {Cambridge},
  year      = {1991},
  isbn      = {978-0-521-40605-5}
}

@inproceedings{gort,
  author    = {Christoph Grunau and Ahmet Alper {\"O}z{\"u}do{\u{g}}ru and
               V{\'a}clav Rozho{\v{n}} and Jakub T{\v{e}}tek},
  title     = {A Nearly Tight Analysis of Greedy k-means++},
  booktitle = {Proceedings of the 2023 Annual {ACM-SIAM} Symposium on
               Discrete Algorithms ({SODA} 2023)},
  pages     = {1012--1070},
  publisher = {Society for Industrial and Applied Mathematics},
  year      = {2023},
  doi       = {10.1137/1.9781611977554.ch39}
}

@article{bahmani,
author = {Bahmani, Bahman and Moseley, Benjamin and Vattani, Andrea and Kumar, Ravi and Vassilvitskii, Sergei},
title = {Scalable k-means++},
year = {2012},
issue_date = {March 2012},
publisher = {VLDB Endowment},
volume = {5},
number = {7},
issn = {2150-8097},
url = {https://doi.org/10.14778/2180912.2180915},
doi = {10.14778/2180912.2180915},
journal = {Proc. VLDB Endow.},
month = mar,
pages = {622–633},
numpages = {12}
}

@inproceedings{makarychev2020,
  author = {Makarychev, Konstantin and Reddy, Aravind and Shan, Liren},
  booktitle = {Advances in Neural Information Processing Systems},
  editor = {H. Larochelle and M. Ranzato and R. Hadsell and M.F. Balcan and H. Lin},
  pages = {16142--16152},
  publisher = {Curran Associates, Inc.},
  title = {Improved Guarantees for k-means++ and k-means++ Parallel},
  url = {https://proceedings.neurips.cc/paper_files/paper/2020/file/ba304f3809ed31d0ad97b5a2b5df2a39-Paper.pdf},
  volume = {33},
  year = {2020}
}

@InProceedings{bachem2016approximate,
  author = {Bachem, Olivier and Lucic, Mario and Hassani, S. Hamed
            and Krause, Andreas},
  title = {Approximate K-Means++ in Sublinear Time},
  booktitle = {Proceedings of the Thirtieth AAAI Conference
               on Artificial Intelligence},
  pages = {1459--1467},
  year = {2016},
  publisher = {AAAI Press},
  doi = {10.1609/aaai.v30i1.10259},
  url = {https://ojs.aaai.org/index.php/AAAI/article/view/10259}
}

@InProceedings{bachem2016fast,
  author = {Bachem, Olivier and Lucic, Mario and Hassani, Hamed
            and Krause, Andreas},
  title = {Fast and Provably Good Seedings for k-Means},
  booktitle = {Advances in Neural Information Processing Systems},
  editor = {Lee, Daniel and Sugiyama, Masashi and Luxburg, Ulrike
            and Guyon, Isabelle and Garnett, Roman},
  publisher = {Curran Associates, Inc.},
  volume = {29},
  year = {2016},
  url = {https://proceedings.neurips.cc/paper_files/paper/2016/hash/d67d8ab4f4c10bf22aa353e27879133c-Abstract.html}
}

@InProceedings{cohenaddad2020fast,
  author = {Cohen-Addad, Vincent and Lattanzi, Silvio
            and Norouzi-Fard, Ashkan and Sohler, Christian
            and Svensson, Ola},
  title = {Fast and Accurate k-means++ via Rejection Sampling},
  booktitle = {Advances in Neural Information Processing Systems},
  editor = {Larochelle, Hugo and Ranzato, Marc'Aurelio and Hadsell, Raia
            and Balcan, Maria-Florina and Lin, Hsuan-Tien},
  pages = {16235--16245},
  publisher = {Curran Associates, Inc.},
  volume = {33},
  year = {2020},
  url = {https://proceedings.neurips.cc/paper/2020/hash/babcff88f8be8c4795bd6f0f8cccca61-Abstract.html}
}

@inproceedings{
shah2026fast,
title={Fast k-means Seeding Under The Manifold Hypothesis},
author={Poojan Chetan Shah and Shashwat Agrawal and Ragesh Jaiswal},
booktitle={Forty-third International Conference on Machine Learning},
year={2026},
url={https://openreview.net/forum?id=lgsMOTIAYu}
}

@inproceedings{aggarwal2009,
author = {Aggarwal, Ankit and Deshpande, Amit and Kannan, Ravi},
title = {Adaptive Sampling for k-Means Clustering},
year = {2009},
isbn = {9783642036842},
publisher = {Springer-Verlag},
address = {Berlin, Heidelberg},
url = {https://doi.org/10.1007/978-3-642-03685-9_2},
doi = {10.1007/978-3-642-03685-9_2},
booktitle = {Proceedings of the 12th International Workshop and 13th International Workshop on Approximation, Randomization, and Combinatorial Optimization. Algorithms and Techniques},
pages = {15–28},
numpages = {14},
location = {Berkeley, CA},
series = {APPROX '09 / RANDOM '09}
}

\end{document}